\documentclass[12pt,draftnofoot,onecolumn]{IEEEtran_icc}%[10pt]
\usepackage{amsmath,amssymb,amsthm}
\usepackage{graphicx,color}
\newtheorem{theorem}{Theorem}
\newtheorem{proposition}{Proposition}
\newtheorem{mydef}{Definition}
\newcommand{\kit}[1]{#1_{t}^{(i)}}
\newcommand{\kitt}[1]{#1_{ii}(t)}
\newcommand{\ki}[1]{#1^{(i)}}
\newcommand{\mc}[1]{\mathcal{#1}}
\newcommand{\mb}[1]{\mathbb{#1}}
\title{Learning Equilibria of a Stochastic Game on Gaussian Interference Channels with Incomplete Information}
\author{\IEEEauthorblockN{Krishna Chaitanya A, Vinod Sharma, and Utpal Mukherji\footnote{Part of this paper is presented in IEEE International Workshop on Signal Processing Advances in Wireless Communications 2015}}\\
\IEEEauthorblockA{Department of ECE, Indian Institute of Science, Bangalore-560012 \\ Email: $\lbrace$akc, vinod, utpal$\rbrace$ @ece.iisc.ernet.in}
}
\date{}
\begin{document}
\maketitle
\begin{abstract}
We consider a wireless communication system in which $N$ transmitter-receiver pairs want to communicate with each other.  Each transmitter transmits data at a certain rate using a power that depends on the channel gain to its receiver.  If a receiver can successfully receive the message, it sends an acknowledgement (ACK), else it sends a negative ACK (NACK).  Each user aims to maximize its probability of successful transmission.  We formulate this problem as a stochastic game and propose a fully distributed learning algorithm to find a correlated equilibrium (CE).  In addition, we use a no regret algorithm to find a coarse correlated equilibrium (CCE) for our power allocation game.  We also propose a fully distributed learning algorithm to find a Pareto optimal solution.  In general Pareto points do not guarantee fairness among the users, therefore we also propose an algorithm to compute a Nash bargaining solution which is Pareto optimal and provides fairness among users.  Finally, under the same game theoretic setup, we study these equilibria and Pareto points when each transmitter sends data at multiple rates rather than at a fixed rate.  We compare the sum rate obtained at the CE, CCE, Nash bargaining solution and the Pareto point and also via some other well known recent algorithms.
\end{abstract}
{\keywords Interference channel, stochastic game, correlated equilibrium, distributed learning, Pareto point.}
\section{Introduction}
In wireless communications due to broadcast medium, transmissions by a user cause interference to other users.  This reduces the transmission rate and/or increases the transmission error of the other users.  Therefore each user (transmitter) aims to use resources like, power and spectrum efficiently to improve its performance which may conflict with the goals of the other users.  Thus, in this paper, we present a game theoretic approach to obtain optimal power allocations that achieve an equilibrium among the different transmitters.\par
We consider a single wireless channel which is being shared by multiple transmitter-receiver pairs.  It is modeled as an interference channel.  Its power allocation in non-game theoretic setup has been studied in \cite{deng}-\cite{PA} and via game theory in \cite{palomar}-\cite{arxiv}.  Power allocation for parallel Gaussian interference channels is studied in \cite{palomar}.  Convergence of iterated best response for computing Nash equilibrium under some conditions on the channel gains is studied for single antenna systems in \cite{palomar}, \cite{Conv}.  In \cite{MIMO_IWF}, convergence of iterative water-filling for parallel Gaussian interference channels with multiple antennas is studied.  Different algorithms are presented in \cite{async}, and \cite{luca} to compute a Nash equilibrium for parallel Gaussian interference channels.  In \cite{arxiv}, we presented an algorithm to compute a Nash equilibrium for a stochastic game on Gaussian interference channels under channel conditions weaker those in \cite{palomar}.\par
In general, in a wireless communication system, a user may not have a knowledge about the other users' channel states and their power policies.  In such a setup, one needs distributed algorithms which each user can use to achieve optimal policies that require less information about the other users.  Online learning algorithms are such a class of algorithms \cite{game_net}.  Some of these algorithms, for example, fictitious play \cite{L_mixed}, are partially distributed algorithms that require some knowledge about other users' strategies to find an equilibrium of the system.  On the other hand, there exist fully distributed learning algorithms \cite{game_net} which do not need any information about the other users' strategies or payoffs to find an equilibrium. \par
We describe the prior work in the literature of wireless communications that has considered learning for optimal power allocation.  The problem of minimizing energy consumption in point-to-point communication with delay constraints in stochastic and unknown traffic and channel conditions is considered in \cite{R_wc}.  This problem is modeled as a Markov decision process and solved using online reinforcement learning.\par
In \cite{D_mac}, orthogonal multiple access channels are considered.  The problem of power allocation is modeled as a non-cooperative potential game and distributed learning algorithms are proposed.  A learning algorithm for finding a Nash equilibrium for a multiple-input and multiple-output multiple access channel is proposed in \cite{L_mimo}.\par
The problem of minimizing the total transmit power of a parallel Gaussian interference channel subject to a minimum signal to interference plus noise ratio (SINR) is considered in \cite{D_te}.  Fully distributed algorithms based on trial and error are proposed to find a Nash equilibrium and a satisfaction equilibrium.  Learning in wireless networks is also considered in the game theoretic framework in \cite{S_chan}-\cite{Q_fn}.  We refer to \cite{game_net} for more information on game theory and learning algorithms for wireless communications.\par
In \cite{r_ce}, a learning algorithm using regrets is proposed to find a CE in a finite game.  Unlike fully distributed algorithms, this no-regret algorithm requires the knowledge of actions chosen by the other users from each play of the game.  The same authors presented a fully distributed learning algorithm in \cite{rl_ce} that leads to a CE when the players are not aware of the functional form of their utility functions.\par
Fully distributed algorithms to find a Nash equilibrium are developed in \cite{L_eff}-\cite{O_agent}.  The algorithms in \cite{L_eff}, \cite{L_tel} are based on trial and error.  Using this algorithm users approach strategies that play a pure strategy Nash equilibrium for a high portion of time.  For potential and dominance solvable games, reinforcement learning algorithms in \cite{L_rlrd}, \cite{O_agent} converge to a NE.\par
We consider a power allocation game on a wireless interference channel.  It is neither a potential game nor dominance solvable.  Even existence of a pure strategy NE is not guaranteed.  Therefore, we can not use the algorithms in \cite{L_rlrd}, \cite{O_agent} to obtain an equilibrium point.  Thus we propose a variation of the regret matching algorithm to find a CE of the proposed game on the wireless communication system without knowing the strategies chosen by other users.  The algorithms proposed in this paper is fully distributed.\par
Apart from learning a non-cooperative equilibrium, learning algorithms for finding a Pareto point also exist in literature (\cite{A_opt}, \cite{A_learn}).  Such points can substantially outperform a CE.\par
\subsection{Contribution}
We make the following contributions in this paper:
\begin{itemize}
\item We propose a fully distributed regret matching algorithm \cite{r_ce} that finds a correlated equilibrium (CE) of the interference channel.  The usual regret matching algorithm is a partially distributed algorithm which requires knowledge of the strategies of the other users.  We propose a modification of that algorithm to convert it into a fully distributed algorithm.  We also compare the sum rate at the CE obtained by our algorithm with that obtained from the algorithm in \cite{rl_ce} and we note that our algorithm converges faster than the algorithm in \cite{rl_ce}.
\item We use a fully distributed no-regret dynamics to compute a CCE of our power allocation game.  In general, every CE is a CCE but the converse may not be true.
\item We propose a fully distributed learning algorithm to find a Pareto point for our game and compare its sum rates with that of a CE.
\item Even though Pareto points outperform CE, and CCE, fairness among users is not guaranteed at a Pareto point.  Using a minor variation of the proposed algorithm to compute a Pareto point, we compute a Nash bargaining solution which guarantees fairness among users.
\item Later, we show that we can use the proposed algorithms to compute CE, CCE, Pareto point and Nash bargaining solutions when each transmitter sends at multiple rates rather than at a fixed rate.
\end{itemize}\par
This paper is organized as follows.  In Section \ref{sm}, we describe the system model and define the problem in game theoretic framework.  We propose and analyse a learning algorithm to find a CE in Section \ref{l_ce}.  For our game, we find Pareto points through a fully distributed learning algorithm in Section \ref{pt}.  We study Nash bargaining solutions for our game in Section \ref{nb}.  We use multiplicative weight algorithm to find CCE of our game in Section \ref{cce}.  In Section \ref{m_rates}, we extend the algorithms to the power allocation problem where each transmitter can transmit at multiple rates.  In Section \ref{ne}, we compare the sum of utilities of all the users at a CE and at a Pareto point and also with other algorithms also for some numerical examples.  Section \ref{concl} concludes the paper.
\section{System Model}\label{sm}
We consider a wireless channel being shared by $N$ independent transmitter-receiver pairs.  Transmission from each transmitter causes interference at other receivers.  The transmitted signal from every transmitter undergoes fading.  The fading gain experienced by the intended signal at a receiver from its corresponding transmitter is called the direct link channel gain.  Similarly, the fading gain experienced by other unintended signals at a receiver is called the cross link channel gain.  We model this scenario as a Gaussian interference channel with fading, where each receiver perceives the transmitted signal with additive white Gaussian noise. \par
Let $\mathcal{H}_d^{(i)} = \{ h_1^{(i)},\dots,h_{n_i}^{(i)}\}$ be the direct link channel gain alphabet and $\mathcal{H}_c^{(i)} = \{ g_1^{(i)},\dots,g_{s_i}^{(i)}\}$ be the cross link channel gain alphabet of user $i \in \{1,\dots,N\}$.  Let the random variable $H_{ij}(t)$ denote the channel gain from transmitter $j$ to receiver $i$ in time slot $t$ which is assumed to be a constant during the slot.  Observe that $H_{ii}(t) \in \mathcal{H}_d^{(i)}$, and $H_{ij}(t) \in \mathcal{H}_c^{(i)}$ for $i\neq j$.  We denote a realization of $H_{ij}(t)$ by $h_{ij}(t)$.  We assume that for a fixed $i,j \in \{1,2,\dots,N\},$ the random variables $H_{ij}(t), t = 1,2,\dots$, are independent and identically distributed.  We also assume that $H_{ij}(t)$ are statistically independent for any $i,j \in \{1,2,\dots,N\}$, and $t = 1,2,\dots$. \par
We assume that transmitter $i$ knows $H_{ii}(t)$ in the begin of slot $t$ but not $H_{ij}(t)$, $i \neq j$; in fact it does not know the distribution of $H_{ij}(t)$ also.  We also assume that transmitter $i$ has finite power levels $p_1^{(i)},\dots,p_{m_i}^{(i)}$ to transmit in a slot.  This is a typical wireless scenario.  For example if, a receiver is sending an ACK/NACK to its transmitter and it is a time duplex channel, then a transmitter can estimate its direct link channel gain but will not know the cross link channel gains;  nor will it know the transmit powers used by the other transmitters.\par
User $i$ transmits $r_i$ bits in every channel use at a power level which depends on the direct link channel gain.  If receiver $i$ successfully receives the message sent in that slot, it sends an ACK (positive acknowledgment), else the receiver sends back an NACK (negative acknowledgement) at the end of the slot.  We assume that ACK messages are small and sent at a low rate such that these are received with negligible probability of error and its transmission overheads are ignored.  This is typically assumed \cite{tcp}.  For a Gaussian channel, the probability of error is a function of the received SINR and the modulation and coding used.  For a given coding and modulation, we can fix a minimum SINR such that the probability of error is negligible above this SINR.  To be specific, we assume that if in time slot $t$, user $i$ transmitted $r_i$ bits per channel use at a power level $p_{l}^{(i)},$ $l \in \{1,2,\dots,m_i\}$ and $H_{ij}(t) = h_{ij}(t)$ then, transmitter $i$ receives an ACK from its corresponding receiver if and only if
\begin{equation*}
  r_i \leq \text{log}\Big(1 + \frac{|h_{ii}(t)|^2p_{l_i}^{(i)}}{1+\Gamma_i\sum_{j\neq i}|h_{ij}(t)|^2p_{l_j}^{(j)}}\Big),  \text{ i.e., }\Gamma_i\sum_{j\neq i}|h_{ij}(t)|^2p_{l_j}^{(j)} \leq \frac{|h_{ii}(t)|^2p_{l_i}^{(i)}}{2^{r_i} -1}-1,
\end{equation*}
and the transmitter receives a NACK, otherwise; where $\Gamma_i$ is a constant that depends on the modulation and coding used by the $i$th receiver.  In the following we will take $\Gamma_i = 1$ for all $i$ for convenience.\par
We consider stationary policies, i.e., the power used by user $i$ in slot $t$ depends only on the channel gain $H_{ii}(t)$ but is independent of time $t$.  Thus, we define the feasible action space of user $i$ by 
\begin{equation}
\mathcal{A}^{(i)} = \left\{{\bf P}^{(i)}=(P^{(i)}_{1},\dots,P^{(i)}_{n_i})|P^{(i)}_l \in \{p_1^{(i)},\dots,p_m^{(i)}\},\right. \left. \sum_{l=1}^{n_i}\pi^{(i)}(l)P^{(i)}_l \leq \overline{P}_i \right\},
\end{equation}
where $\pi^{(i)}$ is the probability distribution of $H_{ii}(t)$.  We note that user $i$ has an average power constraint $\overline{P_i}$ for each feasible action.  As the set of power levels of each user is finite, the number of elements in $\mathcal{A}^{(i)}$ is also finite.   Let the cardinality of $\mathcal{A}^{(i)}$ be $\vert \mathcal{A}^{(i)} \vert = L_i$.  We enumerate the elements of $\mathcal{A}^{(i)}$ as $\{1, 2, \dots, L_i\}$, i.e., when we write $k^{(i)} \in \mathcal{A}^{(i)}$, we mean $k^{(i)}$ is a feasible power policy of user $i$ for $k^{(i)} = 1,\dots,L_i$, and $k^{(i)}(h), h \in \mathcal{H}_d^{(i)}$ denotes the power used when the direct link channel gain is $h$ under the policy $k^{(i)}$.  We denote the action space of the set of all users by $\mathcal{A} = \mathcal{A}^{(1)} \times \dots \times \mathcal{A}^{(N)}$, and an action profile of all users by $k = (k^{(1)}, \dots, k^{(N)})$.  We denote the action space of all users other than user $i$ by $\mathcal{A}^{(-i)} = \mathcal{A}^{(1)} \times \dots \times \mathcal{A}^{(i-1)} \times \mathcal{A}^{(i+1)} \times \dots \times \mathcal{A}^{(N)}$ and the action profile of all users other than $i$ by $k^{(-i)} \in \mathcal{A}^{(-i)}$.  Let $k^{(i)}_{t}$ indicate the action of user $i$ at time $t$.  Also, $k_t = (\kit{k},i=1,\dots,N)$.  A strategy $\phi_i$ of user $i$ is a probability distribution on $\mathcal{A}^{(i)}$, and a pure strategy is a degenerate probability distribution where a certain action is chosen with probability one.  \par
In a given time slot $t$, each user chooses an action that maximizes its probability of successful transmission.  The strategy of each user influences the probability of successful transmission of every user and hence we are interested in finding an equilibrium point.  To model this as a game, we define the reward of user $i$ for a given action profile $(k^{(i)},k^{(-i)})$ in time slot $t$ with direct link channel gain $h_{ii}(t)$ as 
\begin{equation}\label{def_w}
\kit{w}(\kit{k};\kitt{h}) = 
\begin{cases}
1, \text{ if transmitter } i \text{ receives ACK}, \\
0, \text{ else}.
\end{cases}
\end{equation}
This reward of user $i$ in a given time slot $t$ is random as it depends on the cross link channel gains $h_{ij}(t)$ and the power levels $k^{(-i)}_t$ at which the other users transmit in that slot, which in turn depend on the direct link channel gains of those users.  The average reward of user $i$ for the action profile $(k^{(i)},k^{(-i)})$ is 
\begin{equation}
u^{(i)}(k^{(i)},k^{(-i)}) = \limsup\limits_{T \rightarrow \infty} \frac{1}{T} \sum_{t = 1}^T w_{t}^{(i)}(k^{(i)};\kitt{H}).
\end{equation}
By the strong law of large numbers this limit exists $a.s.$ and average reward $u^{(i)}(k^{(i)},k^{(-i)})$ can be interpreted as the probability of successful transmission.  The average reward of user $i$ given the \emph{mixed} strategy $(\phi_i, \phi_{-i})$ is 
\begin{equation}
u^{(i)}(\phi_i, \phi_{-i}) = \sum_{k \in \mathcal{A}} \left[\Pi_{j=1}^N\phi_j(k^{(j)})\right] u_i(k^{(i)},k^{(-i)}).
\end{equation}
Each player aims to maximize its own average reward or probability of successful transmission and we model this scenario as a stochastic game and restrict ourselves to stationary policies.  Then the utility of user $i$ can be written as 
\begin{equation}
u^{(i)}(k^{(i)},k^{(-i)}) = \mathbb{E}\left[ w_{t}^{(i)}(k^{(i)};H_{ij})\right],\label{ex_ut}
\end{equation}
where the expectation is with respect to the distribution of random variables $H_{ij}$ for all $i,j$.  Thus the stochastic game is equivalent to the one-shot game in which user $i$ maximizes its utility defined in (\ref{ex_ut}) and the set of correlated equilibria (CE), defined below, for both of these games is same.
\begin{mydef}
  Given a strategic game $\mathcal{G}$, a joint probability distribution $\phi(k)$ is said to be a {\bf correlated equilibrium} if for all $i = 1,\dots,N$, $\ki{k},\ki{\hat{k}} \in \ki{\mathcal{A}}$ and $k^{(-i)} \in \mathcal{A}^{(-i)}$, we have
  \begin{equation}
    \sum_{k^{(-i)} \in \mathcal{A}^{(-i)}} \phi(k)\left[ u_i(\hat{k}^{(i)},k^{(-i)}) - u_i(\ki{k},k^{(-i)})\right] \leq 0.\label{def_ce}
  \end{equation}
\end{mydef}
We get a correlated $\epsilon$-equilibrium if the zero in the above definition is replaced by $\epsilon$.\par
To find a CE of a one-shot game, a regret matching algorithm is proposed in \cite{r_ce}.  Regret for user $i$ is defined in terms of utility $u^{(i)}$ and it is assumed that the functional form of utility $u^{(i)}$ is known to user $i$.  In this paper, we assume that user $i$ is not aware of functional form of $u^{(i)}$ but knows $w_{t}^{(i)}(k^{(i)};\kitt{h})$ for each $t$ at the end of slot $t$.  We define regret in terms of $w_{t}^{(i)}(k^{(i)};\kitt{H})$ and this definition is equivalent to the definition of regret in \cite{r_ce} because,
\begin{eqnarray}
u^{(i)}(k^{(i)},\phi_{-i}) & = & \sum_{k^{(-i)} \in \mathcal{A}_{-i}}\phi_{-i}(k^{(-i)}) u_i(k^{(i)},k^{(-i)}),\nonumber\\
& = & \sum_{k^{(-i)} \in \mathcal{A}_{-i}} \phi_{-i}(k^{(-i)}) \limsup\limits_{T \rightarrow \infty} \frac{1}{T} \sum_{t = 1,k_{t}^{(-i)} = k^{(-i)}}^T w_{t}^{(i)}(k_t^{(i)};\kitt{H}),\label{x_bf}\\
& = & \limsup\limits_{T \rightarrow \infty} \frac{1}{T} \sum_{k^{(-i)} \in \mathcal{A}_{-i}} \phi_{-i}(k^{(-i)}) \sum_{t = 1}^T w_{t}^{(i)}(k_t^{(i)};\kitt{H}),\label{xchange}\\
& = & \limsup\limits_{T \rightarrow \infty} \frac{1}{T} \sum_{t = 1}^T w_{t}^{(i)}(k_t^{(i)};\kitt{H}).\nonumber
\end{eqnarray}
Note that as the first summation in (\ref{x_bf}) is finite, we can exchange summation and limit to get (\ref{xchange}).  Thus, 
\begin{equation}
u^{(i)}(\ki{k},\hat{k}^{(i)},\phi_{-i}) - u^{(i)}(k^{(i)},\phi_{-i}) = \limsup\limits_{T \rightarrow \infty} \frac{1}{T} \sum_{t = 1}^T \left[w_{t}^{(i)}(\ki{k},\hat{k}^{(i)};\kitt{h}) - w_{t}^{(i)}(k^{(i)};\kitt{h})\right],\label{reg}
\end{equation}
where $u^{(i)}(\ki{k},\hat{k}^{(i)},\phi_{-i})$ is the utility of user $i$, that it would have received if the action $\ki{k}$ is replaced by $\ki{\hat{k}}$ whenever it is played.  Similarly, we write $w_{t}^{(i)}(\ki{k},\hat{k}^{(i)};\kitt{H})$ to denote the reward of user $i$, by playing $\ki{\hat{k}}$ instead of $\ki{k}$.  We use the difference $X_{t}(\hat{k}^{(i)}) = w_{t}^{(i)}(\ki{k},\hat{k}^{(i)};\kitt{H}) - w_{t}^{(i)}(k^{(i)};\kitt{H})$ to define regret. \par
Thus the regret-matching algorithm in \cite{r_ce} can be described as follows. The regret, 
\begin{equation}
R_T(k^{(i)},\hat{k}^{(i)})  =  \max \left\{0, \frac{1}{T}\sum_{t=1}^TA_{t}(k^{(i)},\hat{k}^{(i)})\right\},\label{regret_hm}
\end{equation}
where
\begin{equation}
A_{t}(k^{(i)},\hat{k}^{(i)}) = 
\begin{cases}
X_{t}(\hat{k}^{(i)}), \text{ if }\kit{k} = k^{(i)},\\
0, \text{ else},
\end{cases}
\end{equation}
\begin{equation}
\kit{w}(\hat{k}^{(i)};\kitt{H}) = 
\begin{cases}
1, \text{ if } \kit{\Gamma} \leq \gamma^{(i)}(\hat{k}^{(i)};\kitt{H}),\\
0, \text{ else},
\end{cases}
\label{inst_reg}
\end{equation}
and $\Gamma^{(i)}_{t}$ denotes the actual interference experienced by user $i$ at time $t$ and $\gamma^{(i)}(\kit{k};\kitt{h})$ is the upper bound on interference for successful transmission of user $i$, and it is given by
\begin{equation*}
\gamma^{(i)}(\kit{k};\kitt{h}) = \frac{|\kitt{h}|^2\kit{k}(\kitt{h})}{2^{r_i} -1}-1.
\end{equation*}\par
For every pair of pure strategies $\ki{k}, \ki{\hat{k}} \in \ki{\mathcal{A}}$, the regret $R_T(k^{(i)},\hat{k}^{(i)})$ of a user is a nonnegative real number that reflects the change in utility received by the user if the choice of pure strategy $\ki{k}$ is replaced with $\ki{\hat{k}}$ at every time instance that the user chose to play $\ki{k}$ upto time $T$.  User $i$ chooses a pure strategy according to a probability distribution in which the probability of choosing an action is proportional to its regret.  Thus, user $i$ chooses a pure strategy according to the distribution
\begin{equation}
\phi_{T+1}^{(i)}(\hat{k}^{(i)}) = 
\begin{cases}
\frac{1}{\mu}R_T(k^{(i)},\hat{k}^{(i)}), \text{ for } \hat{k}^{(i)} \neq k^{(i)}, \\
1- \sum_{l\neq k^{(i)}}\frac{1}{\mu}R_T(k^{(i)},l), \text{ if } \hat{k}^{(i)} = k^{(i)},\label{strgy1}
\end{cases}
\end{equation}
for a sufficiently large $\mu$.\par
It is shown in \cite{r_ce} that following the above procedure, the empirical frequencies of actions converge to the set of correlated equilibria.\par
To implement this algorithm each user $i$ not only needs to know its own actions (transmit powers) but also of the other users' actions and its cross link channel gains (to compute regret), which it does not know.  Thus, in the next section we modify this algorithm to find a correlated equilibrium where each user updates its strategy based on the rewards it received and the actions it chose in the past.\par
\section{Learning Algorithm to find a Correlated Equilibrium}\label{l_ce}
The learning algorithm we propose is fully distributed in the sense that every user updates its strategy based on its own actions and rewards and independently of the other users' strategies and rewards.  We will show that the joint empirical distribution converges to the set of correlated equilibria with probability 1 for our algorithm.\par
In our problem as the transmitter is not aware of the interference at the corresponding receiver, it can not find the regrets as in (\ref{regret_hm}).  Therefore, each transmitter estimates regret by estimating the instantaneous reward based on the feedback it has received.  The estimated reward is a function of the strategy $\ki{\hat{k}}$, with respect to which we want to find the regret for not using $\ki{\hat{k}}$ instead of $\ki{k}$.\par
If $\kit{k}$ is the pure strategy that is actually chosen at a time $t$ and $\kitt{h}$ is the direct link channel gain at that time, then the actual interference perceived by its receiver is less than the threshold $\gamma^{(i)}(\kit{k};\kitt{h})$ whenever the communication is successful and it is greater than $\gamma^{(i)}(\kit{k};\kitt{h})$ whenever the communication is a failure.  The user $i$ is optimistic in estimating the rewards for using $\ki{\hat{k}}$ instead of $\ki{k}$.  To formally define the estimated reward we use the following notation.  For each $k^{(i)} \in \mathcal{A}^{(i)}$ and $h_{ii} \in \mathcal{H}_d^{(i)}$,
\begin{eqnarray*}
\mathcal{S}^{(i)}_L(t) & = & \{\hat{k}^{(i)} \in \mathcal{A}^{(i)}: \gamma^{(i)}(\kit{k};\kitt{h}) < \gamma^{(i)}(\hat{k}^{(i)};\kitt{h})\},\\
\mathcal{S}^{(i)}_G(t) & = & \{\hat{k}^{(i)} \in \mathcal{A}^{(i)}: \gamma^{(i)}(\kit{k};\kitt{h}) > \gamma^{(i)}(\hat{k}^{(i)};\kitt{h})\},\\
\mathcal{S}^{(i)}_E(t) & = & \{\hat{k}^{(i)} \in \mathcal{A}^{(i)}: \gamma^{(i)}(\kit{k};\kitt{h}) = \gamma^{(i)}(\hat{k}^{(i)};\kitt{h}\}.
\end{eqnarray*}
User $i$ finds the instantaneous reward $\kit{\tilde{w}}(\ki{k},\hat{k}^{(i)};\kitt{h})$ that could have been if user $i$ had used action $\hat{k}^{(i)}$ instead of $\kit{k}$ at time $t$, as
\begin{equation}
  \kit{\tilde{w}}(\ki{k},\hat{k}^{(i)};\kitt{h}) = 
  \begin{cases}
    1, \text{ if }\kit{w}(\kit{k};\kitt{h}) = 1, \hspace{2.0cm} \hat{k}^{(i)} \in \mathcal{S}^{(i)}_L(t)\cup \mathcal{S}^{(i)}_E(t)\\
    1, \text{ if }\kit{w}(\kit{k};\kitt{h}) = 0, \hspace{2.0cm} \hat{k}^{(i)} \in \mathcal{S}^{(i)}_L(t), \\
    1, \text{ if }\kit{w}(\kit{k};\kitt{h}) = 1, \hspace{2.0cm} \hat{k}^{(i)} \in \mathcal{S}^{(i)}_G(t),\\
    0, \text{ if }\kit{w}(\kit{k};\kitt{h}) = 0, \hspace{2.0cm}  \hat{k}^{(i)} \in \mathcal{S}^{(i)}_G(t)\cup \mathcal{S}^{(i)}_E(t).
  \end{cases}\label{est_reward}
\end{equation} \par  
For every pair of actions $k^{(i)} \text{ and }\hat{k}^{(i)}$, after $T$ slots, the regret $\tilde{R}_T(k^{(i)},\hat{k}^{(i)})$ is 
\begin{equation}
\max\left\{0, \frac{1}{T}\sum_{t=1}^T\tilde{A}_{t}(k^{(i)},\hat{k}^{(i)})\right\},\label{est_regret}
\end{equation}
where
\begin{equation}
\tilde{A}_{t}(k^{(i)},\hat{k}^{(i)}) = 
\begin{cases}
\tilde{X}_{t}(\hat{k}^{(i)}), \text{ if }\kit{k} = k^{(i)},\\
0, \text{ else};
\end{cases}
\end{equation}
\begin{equation}
 \tilde{X}_{t}(\hat{k}^{(i)}) = \kit{\tilde{w}}(\ki{k},\hat{k}^{(i)};\kitt{h}) - \kit{\tilde{w}}(\kit{k};\kitt{h}),\label{proc_1}
\end{equation}
and $\kit{\tilde{w}}(\ki{k},\kit{k};\kitt{h}) = \kit{w}(\kit{k};\kitt{h})$, the actual reward received by the user.  If $k_T^{(i)} = k^{(i)}$, i.e., $k^{(i)}$ is the action chosen by user $i$ at time instant $T$, then an action $\hat{k}^{(i)} \in \mathcal{A}^{(i)}$ in time slot $T+1$ is chosen with probability,
\begin{equation}
\phi_{T+1}^{(i)}(\hat{k}^{(i)}) = 
\begin{cases}
\frac{1}{\mu}\tilde{R}_T(k^{(i)},\hat{k}^{(i)}), \text{ for } \hat{k}^{(i)} \neq k^{(i)}, \\
1- \sum_{l\neq k^{(i)}}\frac{1}{\mu}\tilde{R}_T(k^{(i)},l), \text{ if } \hat{k}^{(i)} = k^{(i)}.\label{strgy}
\end{cases}
\end{equation}
It should be noted that the quantities $\kit{\tilde{w}}, \tilde{R}_T$ are the \emph{estimated} values of reward and regret.\par
We define the empirical frequency of strategies chosen upto time $T$ as 
\begin{equation}
f_T(s) = \frac{1}{T}\vert \left\{t \leq T:k_{t} = \hat{k}\right\} \vert, \hat{k} \in \mathcal{A}.
\end{equation}
It is shown in \cite{r_ce} that the empirical frequency of strategies converges to the set of the correlated $\epsilon$-equilibria if and only if the actual regret $R_T(k^{(i)},\hat{k}^{(i)})$ converges to zero as $T \rightarrow \infty$.  This can be formally stated as
\begin{proposition}\label{prop_ce}
Let $\{k_{t}\}, t = 1, 2,\dots$ be a sequence of actions chosen by the users.  For any $\epsilon \geq 0$, $\limsup\limits_{T \rightarrow \infty}R_T(k^{(i)},\hat{k}^{(i)}) \leq \epsilon$ for each user $i$ and every $k^{(i)}, \hat{k}^{(i)} \in \mathcal{A}^{(i)}$ with $k^{(i)}\neq \hat{k}^{(i)}$, if and only if the sequence of empirical frequencies $f_T$ converges to the set of correlated $\epsilon$-equilibrium almost surely.\qed
\end{proposition}
In Proposition \ref{est_act}, we prove that if the estimated regret $ \tilde{R}_T(k^{(i)},\hat{k}^{(i)}) $ converges to zero then the actual regret $R_T(k^{(i)},\hat{k}^{(i)})$ also converges to zero.
\begin{proposition}\label{est_act}
Let $\{k_{t}\}, t = 1, 2,\dots$ be a sequence of actions chosen by the users.  For each user $i$ and every $k^{(i)},\hat{k}^{(i)} \in \mathcal{A}^{(i)}$ with $k^{(i)}\neq \hat{k}^{(i)}$, if $\lim\limits_{T \rightarrow \infty} \tilde{R}_T(k^{(i)},\hat{k}^{(i)}) = 0,$ then $ \lim\limits_{T \rightarrow \infty} R_T(k^{(i)},\hat{k}^{(i)}) = 0.$
\end{proposition}
\begin{proof}
To prove the proposition, we consider all $t$ such that $\kit{k} = k^{(i)}$ and prove that 
\begin{equation}
\kit{\tilde{w}}(\ki{k},\hat{k}^{(i)};\kitt{h}) \geq \kit{w}(\hat{k}^{(i)};\kitt{h}),\label{step1}
\end{equation}
for any given $k^{(i)},\hat{k}^{(i)} \in \mathcal{A}^{(i)},$ $k^{(i)}\neq \hat{k}^{(i)}$ and given channel gain $\kitt{h}$.  Here, we note that, if $\gamma^{(i)}(k^{(i)};\kitt{h}) = \gamma^{(i)}(\hat{k}^{(i)};\kitt{h})$ then $\kit{\tilde{w}}(\ki{k},\hat{k}^{(i)};\kitt{h}) = \kit{w}(\hat{k}^{(i)};\kitt{h})$.  Therefore (\ref{step1}) is satisfied and hence in the following we consider only the cases where $\gamma^{(i)}(k^{(i)};\kitt{h}) \neq \gamma^{(i)}(\hat{k}^{(i)};\kitt{h})$. \par
We now consider two cases separately: \\
{\it Case} $1$ : $\gamma^{i}(\kit{k};\kitt{h}) < \gamma^{i}(\hat{k}^{(i)};\kitt{h}).$\\ 
In this case, it should be noted that $\kit{\tilde{w}}(\ki{k},\hat{k}^{(i)};\kitt{h}) = 1$ and as $\kit{w}(\hat{k}^{(i)};\kitt{h})$ can be either $0$ or $1$, (\ref{step1}) always holds.\\
{\it Case} $2$ : $\gamma^{i}(\kit{k};\kitt{h}) > \gamma^{i}(\hat{k}^{(i)};\kitt{h}).$\\
In this case, if $\kit{w}(\kit{k};\kitt{h}) = 1$, then by definition $\kit{\tilde{w}}(\ki{k},\hat{k}^{(i)};\kitt{h}) = 1$ and (\ref{step1}) always holds.  If $\kit{w}(\kit{k};\kitt{h}) = 0$, then $\kit{\Gamma} > \gamma^{(i)}(\kit{k};\kitt{h})$ and hence $\kit{\Gamma} > \gamma^{(i)}(\hat{k}^{(i)};\kitt{h})$.  Therefore we have $\kit{w}(\hat{k}^{(i)};\kitt{h}) = 0$ and (\ref{step1}) is satisfied with equality.\par
Hence, (\ref{step1}) always holds and we have $X_{t}(k^{(i)},\hat{k}^{(i)}) \leq \tilde{X}_{t}(k^{(i)},\hat{k}^{(i)})$.  Therefore, $0 \leq R_T(k^{(i)},\hat{k}^{(i)}) \leq \tilde{R}_T(k^{(i)},\hat{k}^{(i)})$, and, if $\tilde{R}_T(k^{(i)},\hat{k}^{(i)})$ converges to zero as $T$ approaches infinity, then $R_T(k^{(i)},\hat{k}^{(i)})$ also converges to zero.
\end{proof}
In \cite{adapt}, authors have extended the regret-matching algorithm of \cite{r_ce} so that one can use a function of the regret in the original procedure instead of regret, where the function satisfies certain conditions.  We can not use that result here, as our estimation does not satisfy the conditions on the function.  But we can generalize the result in \cite{r_ce}.
\begin{theorem}
Let the actual regret be defined as in (\ref{regret_hm}).  Let the actual utility $\kit{w}(\hat{k}^{(i)};\kitt{H})$ in (\ref{inst_reg}) be replaced by an estimated utility $\kit{\tilde{w}}(\ki{k},\hat{k}^{(i)};\kitt{H})$ such that
\begin{equation*}
\kit{\tilde{w}}(\ki{k},\hat{k}^{(i)};\kitt{H}) \geq \kit{w}(\hat{k}^{(i)};\kitt{H}).
\end{equation*}
Then following the regret-matching algorithm (\ref{regret_hm})-(\ref{strgy1}) with the actual regret $R_T(k^{(i)},\hat{k}^{(i)})$ replaced by the estimated regret $\tilde{R}_T(k^{(i)},\hat{k}^{(i)})$, the estimated regret $\tilde{R}_T(k^{(i)},\hat{k}^{(i)})$ converges to zero as $T$ approaches infinity.
\end{theorem}
Following the proof of the main theorem in \cite{r_ce}, we can show that for each $\ki{k}$ and $\ki{\hat{k}}$, the estimated regret $\tilde{R}_T(k^{(i)},\hat{k}^{(i)})$ converges to zeros as $T$ approaches infinity.  Therefore by Proposition \ref{est_act}, we get the following theorem.
\begin{theorem}
If each user chooses strategies in each time slot according to the algorithm (\ref{proc_1})-(\ref{strgy}), then the empirical frequency $f_{T}$ converges to the set of correlated $\epsilon-$equilibria for any $\epsilon > 0$.
\end{theorem}
In the proof of the main theorem of \cite{r_ce}, history up to time $T$ is defined as the actions chosen by all users at time instances $t = 1,\dots,T$.  To prove that the estimated regret converges to zero following the regret-matching algorithm, we just need to redefine the history up to time $T$ as the actions chosen by all users along with the direct channel states at time instances $t = 1,\dots,T$.  With this definition of history, the entire proof of the main theorem in \cite{r_ce}, carries over and we can conclude that the estimated regret converges to zero.\par
The performance of the system at a CE may not be very satisfactory from the overall system point of view.  Therefore, we also provide a distributed algorithm in the next section which achieves a Pareto point.  The Pareto points are socially optimal.
\section{Learning Coarse Correlated Equilibrium}\label{cce}
In this section, we compute a coarse correlated equilibrium which is a generalization of a correlated equilibrium.  We present the multiplicative weight (MW) algorithm \cite{m_w},\cite{plg} to compute a CCE of our power allocation game.  MW algorithm has much less computational complexity per iteration than that of the regret matching algorithm presented in Section \ref{l_ce}.  It also does not require estimation of regret as needed in Section \ref{l_ce}.  Also, it has been observed that the price of anarchy (POA) of a CCE is no worse than that of a CE in a large class of games \cite{poa}.  However, it is also known that for some other classes of games, e.g., congestion games, the POA of CCE/CE can be larger compared to NE.\par
From the definition of CE, condition (\ref{def_ce}) requires that every user minimizes the conditional expectation of utility where the conditioning is on $\phi$ and $\ki{k}$.  In CCE, user $i$ contemplates a deviation $\ki{\hat{k}}$ knowing only the distribution $\phi$. \par
Let $\ki{c}(\ki{k},k^{(-i)})$ be the cost of user $i$ and each user chooses its action to minimize the cost.  In our power allocation problem, we can define cost as negative of the utility, i.e., $\ki{c}(\ki{k},k^{(-i)}) = -\ki{u}(\ki{k},k^{(-i)})$.  We define the CCE of a cost minimization game as
\begin{mydef}
A distribution $\phi$ on $\mc{A}$ is said to be a {\bf coarse correlated equilibrium} if 
\begin{equation}
\mb{E}_{k \sim \phi}\left[\ki{C}(k)\right] \leq \mb{E}_{k \sim \phi}\left[\ki{C}(\ki{\hat{k}},k^{(-i)})\right],
\end{equation}
for each user $i \in \{1,2,\dots,N\}$, and for all actions $\ki{k},\ki{\hat{k}} \in \mc{A}_i$.  The distribution $\phi$ is called a $\epsilon$-coarse correlated equilibrium if 
\begin{equation}
\mb{E}_{k \sim \phi}\left[\ki{C}(k)\right] \leq \mb{E}_{k \sim \phi}\left[\ki{C}(\ki{\hat{k}},k^{(-i)})\right] + \epsilon,
\end{equation}
for each user $i \in \{1,2,\dots,N\}$, for every action $\ki{k},\ki{\hat{k}} \in \mc{A}_i$.
\end{mydef}\par
Please note that whenever the cost is a random variable, we denote it by $\ki{C}$ rather than $\ki{c}$.  In this definition, cost is a random variable that depends on the randomly chosen actions $\ki{k}$.\par
Every CE is also a CCE and thus the set of CCE is a larger set than the set of CE.  There exist no-regret learning algorithms to compute a CCE but the notion of regret used to compute a CCE is different from that used to compute a CE.  The regret defined in Section \ref{sm} is known as internal regret and we use external regret to compute a CCE which is defined as 
\begin{mydef}
The regret of user $i$ given the pure strategy sequence $\ki{k}_1,\dots,\ki{k}_T$ with respect to an action $\ki{k}$ is 
\begin{equation}%\begin{multline}
\frac{1}{T}\sum_{t=1}^T \mb{E}_{k^{(-i)}\sim \phi^{(-i)}}\left[\kit{C}(\ki{k}_t,k^{(-i)}) - \kit{C}(\ki{k},k^{(-i)})\right].\label{extn_regret}
\end{equation}%\end{multline}
\end{mydef}
An algorithm in which users update their strategies based on the received cost in such a way that the external regret converges to zero is a no-regret algorithm.  We now present a no-regret algorithm known as multiplicative weight algorithm to compute a CCE.\par
In the initial iteration $t=1$, each user assigns a weight $\ki{\mu}_t(\ki{k}) = 1$ to action $\ki{k} \in \mc{A}_i$.  User $i$ chooses an action $\ki{k}$ with probability
\begin{equation}
\ki{q}_t(\ki{k}) = \frac{\ki{\mu}_{t}(\ki{k})}{\sum_{\ki{\hat{k}}\in \mc{A}_i} \ki{\mu}_{t}(\ki{\hat{k}})}.\label{prob_update}
\end{equation}
During the iteration $t$, if $\ki{k}$ is the action chosen by user $i$ in iteration $t$, then it receives the expected utility $\ki{c}_t(\ki{k}) = \mb{E}_{k^{(-i)}\sim \phi^{(-i)}}[\ki{C}(\ki{k},k^{(-i)})]$.  Based on the received utility, user $i$ updates the weight $\ki{\mu}_{t}(\ki{k})$ of action $\ki{k}$, as 
\begin{equation}
\ki{\mu}_{t+1}(\ki{k}) = \ki{\mu}_t(\ki{k})(1-\epsilon)^{\ki{c}_t(\ki{k})}.\label{wt_update}
\end{equation}
For the iteration $t+1$, user $i$ chooses an action according to (\ref{prob_update}) with weights $\ki{\mu}_t$ replaced by $\ki{\mu}_{t+1}$ and this process is repeated. We have the following convergence result.
\begin{theorem}
Following the multiplicative weight update algorithm, there exists a positive integer $T$ such that the external regret of user $i$ defined in (\ref{extn_regret}) is less than $\epsilon$ after $T$ iterations.  Let $\phi_t = \Pi_{i=1}^N \ki{p}_t$ denote the outcome distribution at time $t$ and $\phi = \frac{1}{T}\sum_{t=1}^T \phi_t$.  Then $\phi$ is a $\epsilon$-coarse correlated equilibrium.
\end{theorem}
We use this MW algorithm to find a CCE of our power allocation game.  In general to use the MW algorithm user $i$ needs to know the expected utility.  In our game user $i$ finds it given the history of actions and rewards as 
\begin{eqnarray*}
\ki{u}_t(\ki{k}) & = &\mb{E}_{k^{(-i)}\sim \phi^{(-i)}}[\ki{U}_t(\ki{k},k^{(-i)})],\\
& = &\frac{1}{t}\sum_{\tilde{t}=1}^tw^{(i)}_{\tilde{t}}(\ki{k},H_{ii}(\tilde{t}))1_{k^{(i)}_{\tilde{t}} = \ki{k}}.
\end{eqnarray*}
Based on $\ki{u}_t$, user $i$ updates its weights using the MW update and chooses action according to (\ref{prob_update}) with $\ki{c}_t(\ki{k}) = -\ki{u}_t(\ki{k})$.  Unlike in the algorithm of CE, we do not need to evaluate the estimated reward as the MW algorithm does not explicitly depend on the regret defined in (\ref{extn_regret}).  But the MW algorithm guarantees that the external regret converges to zero.  Hence we can apply the MW algorithm to our problem to find a CCE.\par
\section{Pareto Optimal Points}\label{pt}
\begin{mydef}
An action profile $k \in \mathcal{A}$ is {\bf Pareto optimal} if there does not exist another action profile $\hat{k} \in \mathcal{A}$ such that $u_i(\hat{k}) \geq u_i(k)$ for all $i = 1,\dots,N$ with at least one strict inequality.
\end{mydef}
In this section we present a \emph{distributed} algorithm to find a Pareto optimal point.\par
The global maximum of 
\begin{eqnarray}
W(k) & = &\sum_{i=1}^N \alpha_i u_i(k),\nonumber \\
\text{ subject to } & & k \in \mathcal{A},\label{pareto}
\end{eqnarray}
is a Pareto optimal solution, where $\alpha_i$ are positive constants \cite{MOP}.  We find Pareto points by finding a solution of (\ref{pareto}).\par
We assume that when a receiver sends an ACK/NACK to its transmitter, all the other transmitters can also listen to it without error.  This is realistic in many wireless systems because an ACK/NACK message is small and is usually transmitted at low rates with very low probability of error.  Under this assumption, we present a learning algorithm in which users may or may not choose to experiment and update their strategies in such a way that improves $W(k)$.\par
The algorithm is as follows:
\begin{itemize}
\item Each user $i$ chooses a random action $\ki{k}$ uniformly from $\ki{\mathcal{A}}$.  All the users use these randomly chosen actions for a fixed number $T$ of time slots.  Each user $i = 1,\dots,N,$ follows the procedure below sequentially:
\item As user $i$ receives the feedback of other users, it finds the weighted sum $\hat{W}(k)$ of the utilities 
\begin{equation}
\hat{W}(k) = \sum_{i=1}^N \alpha_i \left(\frac{1}{T}\sum_{t = 1}^T  \kit{w}(\ki{k};\kitt{H})\right).
\end{equation}
At the end of $T$ slots, user $i$ experiment with probability $\delta$. When user $i$ experiments, with probability $\epsilon$, chooses an action randomly with uniform probability from $\mathcal{A}^{(i)}$ other than $\ki{k}$, and with probability $1-\epsilon$ it chooses an action other than $\ki{k}$ in the following way:
\begin{itemize}
\item In the action $\ki{k}$, a power level has been specified for each value of direct link channel gain.  User $i$ chooses an action randomly from a subset of $\mathcal{A}^{(i)}$, with feasible actions having higher power level than $\ki{k}$ for a channel state with the highest probability of occurrence.  If this subset is empty, then it chooses an action with higher power level for the channel gain with second highest probability of occurrence. 
\item If all the direct link channel gains occur with equal probability, then user $i$ chooses an action randomly from a subset of $\ki{\mathcal{A}}$, with feasible actions having higher power level for the maximum value of direct link channel gain.  If this subset is empty, it chooses an action with higher power level for the second maximum direct link channel gain.
\end{itemize}
\item Let this new action be $\hat{k}^{(i)}$.  For the next $T$ time slots, user $i$ uses action $\hat{k}^{(i)}$, and user $j$ uses actions $k^{(j)}$ for $j \neq i$.  User $i$ finds the weighted sum of the utilities of all the users $\hat{W}(\hat{k}^{(i)},k^{(-i)})$.  If $\hat{W}(\hat{k}^{(i)},k^{(-i)}) > \hat{W}(k^{(i)},k^{(-i)})$, then user $i$ replaces its action $k^{(i)}$ with $\hat{k}^{(i)}$, and this new weighted sum of the average rewards is taken as a benchmark.  If there is no improvement in the weighted sum of the average rewards, it randomly selects another action following the procedure described above.  Thus each user may experiment with upto a maximum of MAX number of actions chosen randomly.
\end{itemize}\par
If $\epsilon = 1$, each user experiments with randomly chosen actions.  But, for small $\epsilon$, in our algorithm we are selecting an action from the action space of that user by a \emph{local search}.  The local search often yields a better point, that improves $W(k)$, than a purely random search in the entire action space, and yields a faster rate of convergence as seen in our numerical examples.\par
A user updates its action whenever there is an improvement in the weighted sum of average reward over the benchmark.  Hence, this benchmark of utility is monotonically increasing and bounded above by $\sum \alpha_i$.  Therefore, for a sufficiently large $M$, we find a Pareto optimal point with a large probability.  By increasing MAX, this probability can be made arbitrarily close to $1$.\par
Our algorithm is a distributed version of a meta heuristic, \emph{stochastic local search} \cite{s_s}, often used for global optimization.\par
We can also obtain a Pareto point which satisfies certain minimum probability of success (e.g., for voice users) by including this constraint in the set $\mathcal{A}$.  Pareto points, although they globally maximize $W(k)$, may be \emph{unfair} to some users.  Changing the weights $\alpha_i$ can alleviate some unfairness.  Otherwise, one can obtain Pareto points which are \emph{Nash bargaining solutions} \cite{BASAR}, which can be obtained via a similar algorithm as explained in the next section.
\section{Nash Bargaining}\label{nb}
In Nash bargaining, we specify a disagreement outcome that specifies utility of each user that it receives by playing the disagreement strategy whenever there is no incentive to play the bargaining outcome.  Thus, by choosing the disagreement outcomes appropriately, the users can ensure certain fairness.  \par
The Nash bargaining solutions are Pareto optimal and also satisfy certain natural axioms \cite{nash_b}.  It is shown in \cite{nash_b} that for a two player game, there exists a unique bargaining solution (if the feasible region is nonempty) that satisfies the axioms stated above and it is given by the solution of the optimization problem
\begin{eqnarray}
\text{maximize } & & (s_1-d_1)(s_2-d_2), \nonumber \\
\text{subject to } & & s_i \geq d_i, i = 1,2, (s_1,s_2) \in \mc{S}.
\end{eqnarray}
For an N-user Nash bargaining problem, this result can be extended and the solution of an N-user bargaining problem is the solution of the optimization problem
\begin{eqnarray}
\text{maximize } & & \Pi_{i=1}^N (s_i-d_i), \nonumber \\
\text{subject to } & & s_i \geq d_i, i = 1,\dots,N, (s_1,\dots,s_N) \in \mc{S}.\label{nb_n}
\end{eqnarray}\par
A Nash bargaining solution is also related to proportional fairness, another fairness concept commonly used in communication literature.  A utility vector $s^* \in \mc{S}$ is said to be \emph{proportionally fair} if for any other feasible vector $s \in \mc{S}$, for each $k$, the aggregate proportional change $s_k-s^*_k/s_k^*$ is non-positive \cite{prop_fair}.  If the set $\mc{S}$ is convex, then Nash bargaining and proportional fairness are equivalent \cite{prop_fair}.  Proportional fairness is studied in \cite{nb_pf} when $\mc{S}$ is non-convex.  In our case, $\mc{S}$ is convex and hence Nash bargaining solution is also proportionally fair.\par
A major problem in finding a solution of a bargaining problem is choosing the disagreement outcome.  It is more common to consider an equilibrium point as a disagreement outcome.  In our problem we can consider the utility vector at a CE as the disagreement outcome.  We can also choose $d_i = 0$ for each $i$.  If we choose the disagreement outcome to be a CE, each user needs to evaluate a CE first before running the algorithm to find a solution of (\ref{nb_n}), which requires more computations.  Instead, we can choose the disagreement outcome to be the zero vector or by using the following procedure :
\begin{itemize}
\item Each user chooses an action that gives higher power level to the channel state that has higher probability of occurrence.  In other words, among the set of feasible actions, choose a subset of pure strategies that gives the highest power level to the channel state with highest probability of occurrence.  We shrink the subset by considering the actions that give higher power level to the second frequently occurring channel state and we repeat this process until we get a single strategy.
\item If all the channel states occur with equal probability, we follow the above procedure by considering the value of the channel gain instead of the probabilities of occurrence of the channel gains. 
\end{itemize}
Let the pure strategy chosen by user $i$ be $\ki{k}$, and assume that the users use these strategies for a fixed number $T_d$ of slots.  User $i$ finds $d_i$ by averaging the rewards received in the $T_d$ slots, i.e.,
\begin{equation}
d_i = \frac{1}{T_d}\sum_{t=1}^{T_d} w_i(\ki{k};\kitt{H}).
\end{equation} \par
For our numerical evaluations we have chosen the disagreement outcome following the procedure described above instead of choosing the zero vector.  To find the bargaining solution, i.e., to solve the optimization problem (\ref{nb_n}), we use the algorithm of Section \ref{pt} used to find a Pareto optimal point but with objective $W(k)$ defined as
\begin{equation*}
 W(k) = \Pi_{i=1}^N (u_i(k)-d_i).
\end{equation*}
In Section \ref{ne}, we present a Nash bargaining solution for the numerical examples we consider, and observe that the Nash bargaining solution obtained is a Pareto optimal point which provides fairness among the users.
\section{Transmission at Multiple Rates}\label{m_rates}
Until now, we have presented learning algorithms to compute a CE, a CCE, Pareto points and Nash bargaining solutions, when a user is transmitting at a fixed rate.  In this section, we generalize the model so that a user can transmit at multiple rates rather than at a fixed rate and show that we can still use the same algorithms to compute equilibria.\par
Let $\ki{\mc{R}} = \{\ki{r}_1,\dots,\ki{r}_{l_i}\}$ be the set of possible transmission rates of user $i$.  Let $\{p_1^{(i)},\dots,p_{m_i}^{(i)}\}$ be the set of power levels for user $i$, as considered earlier.  We denote the new strategy set as 
\begin{multline}
\ki{\tilde{\mc{A}}} = \left\{(\ki{r},P^{(i)}_{1},\dots,P^{(i)}_{n_i})|\ki{r}\in \ki{\mc{R}}, P^{(i)}_l \in \{p_1^{(i)},\dots,p_m^{(i)}\}, \sum_{l=1}^{n_i}\pi^{(i)}(l)P^{(i)}_l \leq \overline{P}_i \right\}.
\end{multline}
The cardinality of $\ki{\tilde{\mc{A}}}$ is $l_i$ times that of $\ki{\mc{A}}$, as every action $\ki{k} \in \ki{\mc{A}}$ can be associated with each rate in $\ki{\mc{R}}$.  We enumerate the elements of $\ki{\tilde{\mc{A}}}$ as in Section \ref{sm}.  Here also we denote an action by $\ki{k} \in \ki{\tilde{\mc{A}}}$, and $r(\ki{k})$ is the rate of transmission under the action $\ki{k}$.  If $H_{ii} \in \ki{\mc{H}}_d$ is the direct link channel gain of user $i$ and the user chooses action $\ki{k}$, then it transmits at a rate $r(\ki{k})$ with power $\ki{k}(H_{ii})$.\par
User $i$ receives an ACK if the interference at receiver $i$ satisfies
\begin{equation}
\ki{I} \leq \frac{|H_{ii}|^2\ki{k}(H_{ii})}{2^{r(\ki{k})} -1}-1,\label{ub_rate}
\end{equation}
and it receives a NACK otherwise.  We use the same notation $\gamma^{(i)}(\ki{k};H_{ii})$ to denote the upper bound on the interference for receiving an ACK.  We can redefine the estimated reward and estimated regret as in (\ref{est_reward}) and (\ref{est_regret}) respectively, but with the threshold redefined as
\begin{equation}
\gamma^{(i)}(\kit{k};\kitt{H}) = \frac{|\kitt{H}|^2\kit{k}(\kitt{H})}{2^{r(\kit{k})} -1}-1.
\end{equation}\par
It can easily be seen even in this case that the estimated reward $\kit{\tilde{w}}(\hat{k}^{(i)};\kitt{H})$ is greater than or equal to the actual reward $\kit{w}(\hat{k}^{(i)};\kitt{H})$.  Hence, we can use the regret-matching algorithm to compute CE for the game with $\ki{\tilde{\mc{A}}}$ as the strategy set.  We can also use the respective algorithms mentioned earlier to compute Pareto points, Nash bargaining solution, and CCE.\par
We can ensure that these solutions satisfy certain minimum rates by limiting our overall action space to strategies that satisfy these rate constraints.
\section{Numerical Examples}\label{ne}
In this section we consider three examples with three transmitter-receiver pairs in the communication system.  In the first example, we consider a symmetric scenario where $\mathcal{H}_d^{(i)} = \{0.2, 0.6, 1\}$ and $\mathcal{H}_c^{(i)} = \{0.1, 0.3, 0.5\}$ for each user $i$, and each channel state occurs with equal probability.  The set of possible power levels for each player is $\{0, 5, 10, 15, 20, 25, 30\}$.  Each user transmits at a rate of $r_i = 0.75$ bits per channel use, and receives feedback from its receiver.  Each user follows the learning algorithm (\ref{proc_1})-(\ref{strgy}) to find a CE, and finds a Pareto optimal strategy as described in Section \ref{pt}.  In finding the Pareto points, we choose $\alpha_i = r_i$ for all $i$.  The sum rate at a CE and at a Pareto point are compared in Figure \ref{fig_1}.  We also compare the sum rate at a CE obtained by using the reinforcement learning (RL) algorithm in \cite{rl_ce}.  We observe that the sum rates at CE obtained via our algorithm and that obtained via the algorithm in \cite{rl_ce} almost coincide in this example.\par
Even though the sum rates are close for both the algorithms, we observe that our algorithm convergences faster than the RL algorithm.  In the RL algorithm, it is required that each pure strategy of each user should be played for a minimum number of time slots to find the regret as defined in \cite{rl_ce}.  Thus the algorithm requires a larger number of iterations to converge to the set of correlated equilibria.  In this example, at SNR of 15dB, our algorithm converges in about $200,000$ iterations, whereas the algorithm in \cite{rl_ce} converges in about $700,000$ iterations.\par
In finding a Pareto point, if we randomly choose a strategy $(\epsilon = 1)$ each time, instead of local search, the algorithm runs for about $150,000$ iterations whereas our local search algorithm finds a Pareto point in $70,000$ iterations.\par
We also plot in Figure \ref{fig_1} the sum rate at a stochastically stable point of the trial and error based algorithm (TE) in \cite{L_eff}.  It is known from \cite{L_eff} that the algorithm therein converges to an efficient NE only if the game under consideration has at least one pure strategy NE.  In general, we can not guarantee existence of a pure strategy NE for our game, and hence the stochastically stable point computed by the algorithm in \cite{L_eff} need not be a NE.  Then the algorithm produces stochastically stable points that maximize $g(k) = \alpha_1 W(k) - \beta_1 S(k)$ for all $k \in \mathcal{A}$  where $W(k) = \sum_{i=1}^N u_i(k)$, and 
\begin{equation}
S(k) = \text{min}\{\delta: \text{the benchmark actions }\text{constitute a $\delta$-equilibrium}\}.
\end{equation}
We refer to \cite{L_eff} for further details of the function $g$.\par
We also plot the sum rate at the CCE and at the Nash bargaining solution obtained for Example 1 in Figure \ref{fig_1}.  We observe that the sum rate at a CCE is better than that at a CE, but that the MW algorithm runs for about $250,000$ iterations which is more than the number of iterations required for computing a CE.  The sum rate at the Nash bargaining solution is very close to that at the Pareto point, but the former provides fairness among users.  We present the rates at the Pareto point and at the Nash bargaining solution in Table \ref{table1} to illustrate the fairness provided by the Nash bargaining solution for Example 1.  The rates of all the three users are mentioned as a triplet $(r_1, r_2, r_3)$, where $r_i$ is the rate of user $i$.  It can be seen for several SNR values, that the Nash bargaining solution provides more fairness than at the Pareto point.  Even though Example 1 is symmetric, as the algorithms are based on stochastic local search, rate allocations need not be symmetric.\par
\begin{figure}
  \centering
  %\hspace{-0.5cm}
  \includegraphics[height=9.5cm,width=17.0cm]{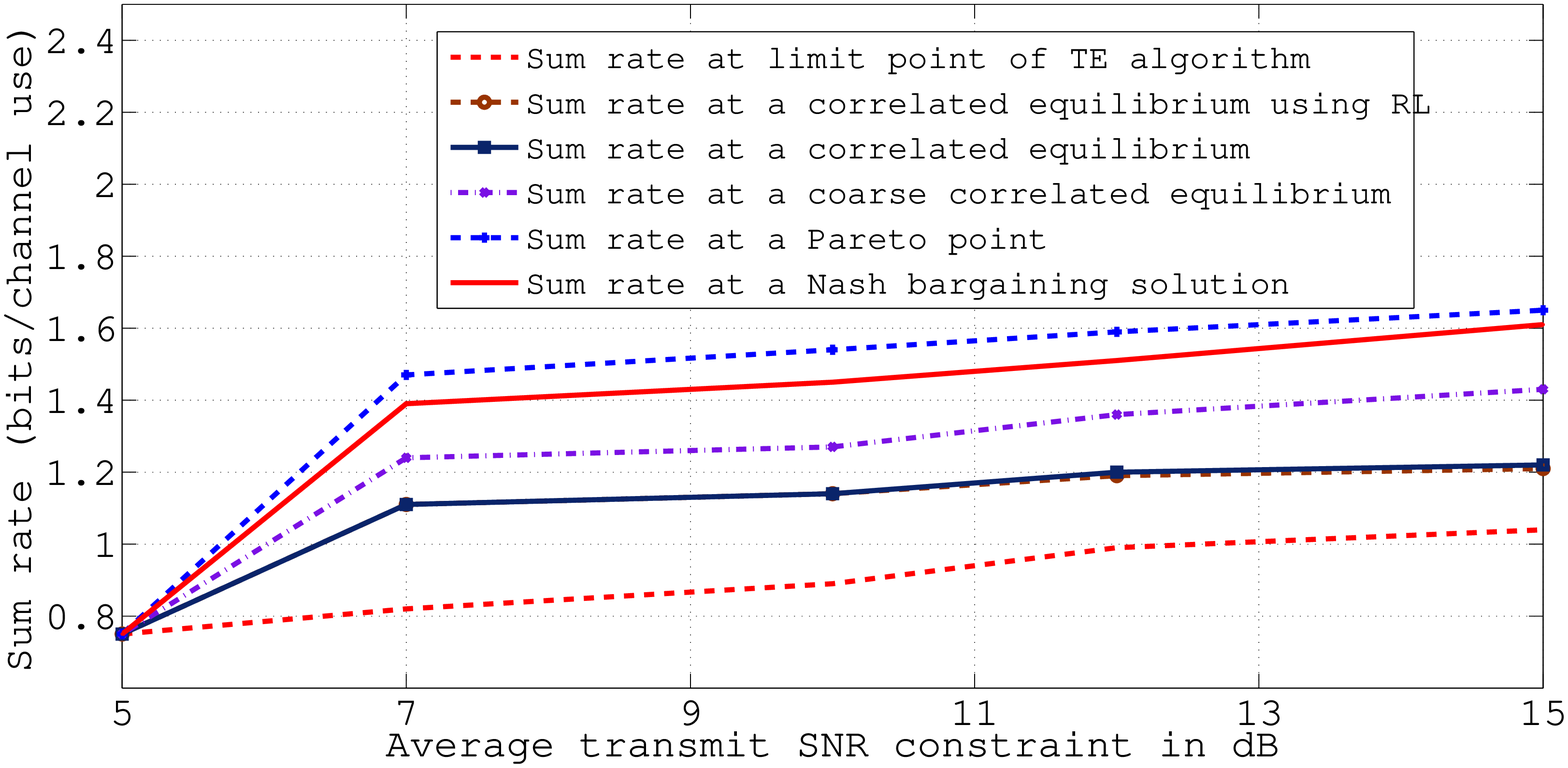}
  %\vspace{-0.2cm}
  \caption{Sum rates at Pareto points and CE for Example 1.} 
  \label{fig_1}
  %\vspace{-0.4cm}
\end{figure}

\begin{table}[h]
  \centering  
  {\footnotesize
  \begin{tabular}{|c|c|c|}
    \hline 
    SNR(dB) & Rates at Pareto point & Rates at Nash bargaining\\
    \hline
    & & \\
    5 & (0.25, 0.25, 0.25) & (0.25, 0.25, 0.25)\\ 
    & & \\    
    7 & (0.67, 0.37, 0.43) & (0.48, 0.45, 0.46)\\
    & & \\    
    10 & (0.58, 0.77, 0.19) & (0.49, 0.49, 0.47)\\
    & &\\    
    12 & (0.36, 0.34, 0.89) & (0.50, 0.51, 0.5)\\
    & &\\   
    15 & (0.55, 0.41, 0.69) & (0.54, 0.52, 0.55)\\    
    & &\\
    \hline
  \end{tabular}
  }
  \vspace{0.3cm}
  \caption{Fairness of rates at the Pareto point and the Nash bargaining solution for Example 1.}
  %\vspace{-0.6cm}
  \label{table1}
\end{table}
We note that the sum rate of all the users is higher at the Pareto optimal point than at a CE.  The improvement is 21.8\% at the average transmit SNR constraint of 10dB, and 24.6\% at the SNR of 15dB. \par
In the computation of CE, for each pure strategy $k^{(i)} \in \ki{\mc{A}}$, we need to estimate the regret, which requires some calculation of a threshold in advance before starting the running of the algorithm.  This requires two multiplications and two additions per action for each user.  But, for one iteration in the MW algorithm, each user requires one division per action and two multiplications.  Hence, even though the regret-matching algorithm requires computation of estimated regret, it may converge faster than the MW algorithm which does not require computation of regret.  It is observed from examples that the regret-matching has relatively less running time than the MW algorithm.\par
Next we consider an asymmetric scenario, in Example 2.  In this example also we consider $\mathcal{H}_d^{(i)} = \{0.2, 0.6, 1\}$ and $\mathcal{H}_c^{(i)} = \{0.1, 0.3, 0.5\}$ for each user $i$.  The direct link gains from $\mathcal{H}_d^{(i)}$ occur with equal probability for each user $i$, but the cross link gains occur with a different probability distribution for each user.  For user $1$, the distribution is $\{0.5, 0.3, 0.2\}$, for user 2 it is $\{0.4, 0.5, 0.1\}$, and for user 3, it is $\{0.25, 0.5, 0.25\}$.  Users $1, 2,$ and $3$ transmit at rates $0.5, 0.75,\text{ and }0.9$ bits per channel use respectively.  The sum rate at the CE and at the Pareto point obtained from our algorithm are compared in Figure \ref{fig_2}.  We also compare the sum rate at the CE obtained by using the RL algorithm in \cite{rl_ce}.  In this example also, we observe that our algorithm converges faster than the RL algorithm: at SNR of 15dB, our algorithm converges in about $250,000$ iterations, whereas the algorithm in \cite{rl_ce} converges in about $770,000$ iterations.  We also plot the sum rate at a stochastically stable point of the algorithm in \cite{L_eff} which maximizes $g(k) = \alpha_1 W(k) - \beta_1 S(k)$. \par
\begin{figure}
  \centering
  %\hspace{-0.5cm}
  \includegraphics[height=10.0cm,width=17.0cm]{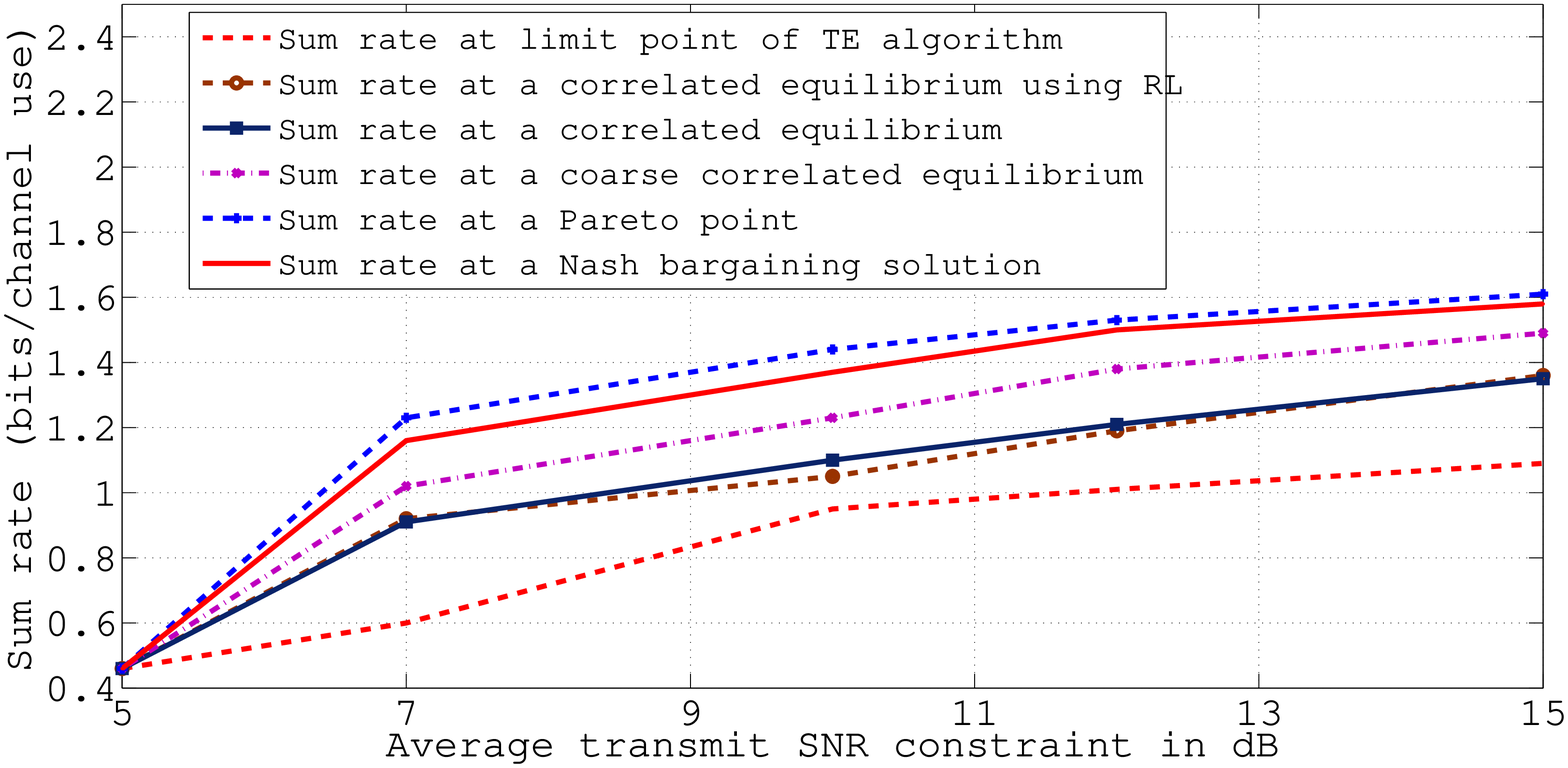}
  %\vspace{-0.2cm}
  \caption{Sum rates at Pareto points and CE for Example 2.}  
  \label{fig_2}
  %\vspace{-0.4cm}
\end{figure}
We also plot the sum rate at a CCE and at a Nash bargaining solution for Example 2 in Figure \ref{fig_2}.  In this example also, we observe an improvement in the sum rate at a CCE over that at a CE, and the MW algorithm runs for about $325,000$ iterations which is more than the number of iterations required for computing a CE.  The sum rate at the Nash bargaining solution is very close to that at the Pareto point.  We observe that the sum of the rates of all the users is higher at the Pareto optimal point than at a CE.  We observe an improvement of 22.7\% at SNR 10dB and an improvement of 17.5\% at SNR of 15dB.  The sum rate obtained via \cite{L_eff} at its stochastically stable point is the lowest.\par
\begin{figure}
  \centering
  %\hspace{-0.5cm}
  \includegraphics[height=10.0cm,width=17.0cm]{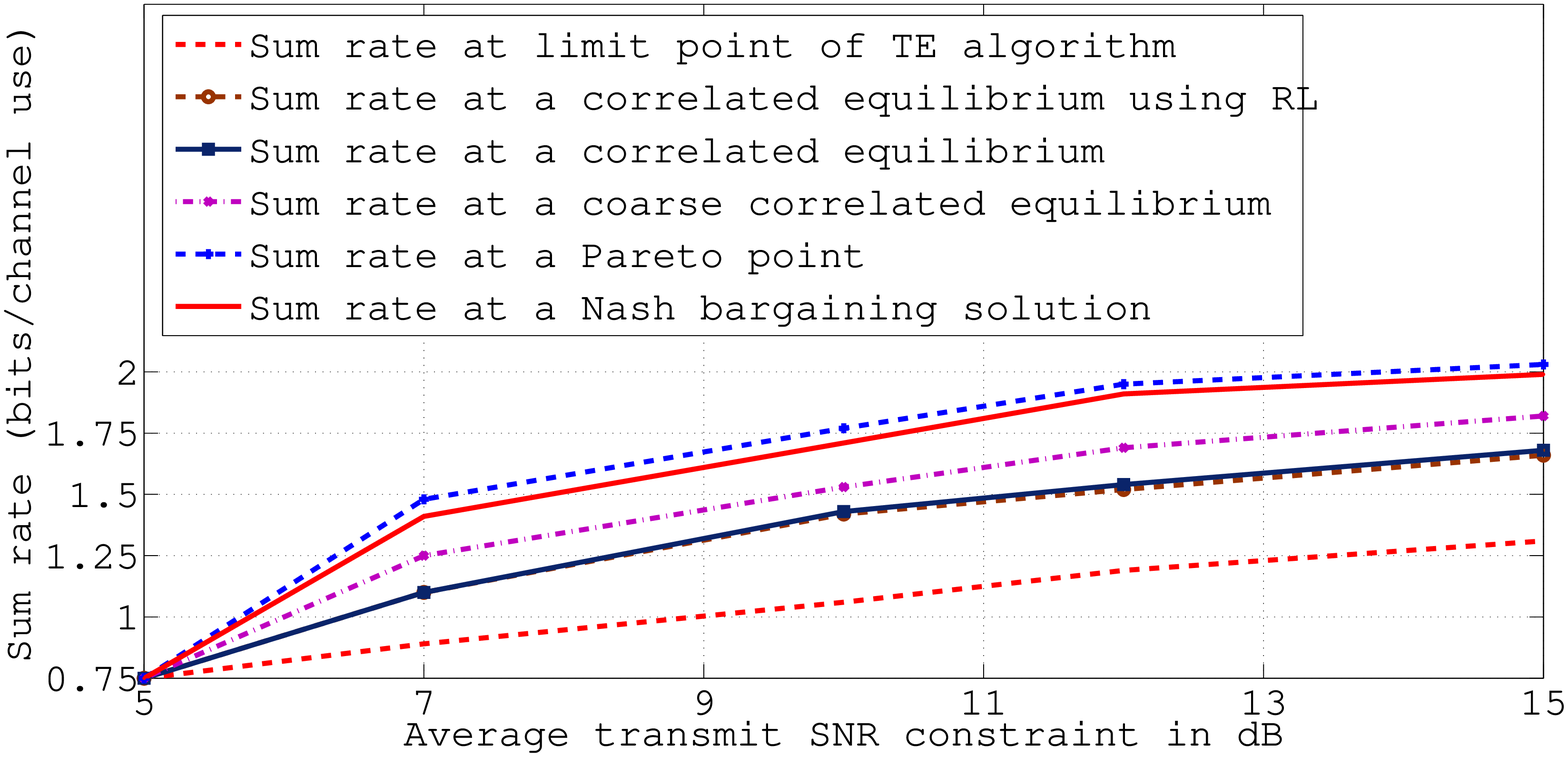}
  %\vspace{-0.2cm}
  \caption{Sum rates at Pareto points and CE for Example 3.} 
  \label{fig_3}
  %\vspace{-0.4cm}
\end{figure}
Finally, we consider multiple rates of transmission in Example 3.  For Example 3, we consider the same parameters as in Example 2, but each user can send data at any rate from the set $\mc{R} = \{0.75, 0.9, 1.2\}$.  We compare the sum rates at a CE, CCE, Pareto point, and Nash bargaining solution, in Figure \ref{fig_3}.  We also plot the sum rates using the RL algorithm and TE algorithm for this example in Figure \ref{fig_3}.  In this example also we observe an improvement in the sum rate at the CCE over the sum rate at the CE.  Sum rate at the Nash bargaining solution and at the Pareto point almost coincide in this example also.  As the cardinality of the strategy set of each user is enlarged by transmitting at multiple rates, the regret-matching algorithm runs for about $400,000$ iterations to compute a CE and the MW algorithm runs for about $475,000$ iterations to compute a CCE, at SNR of 15dB.
\section{Conclusions}\label{concl}
We have considered a communication system in which $N$ transmitter-receiver pairs communicate on a wireless channel.  Each transmitter sends data at a certain rate at a power level that is a function of the direct link channel gain and the feedback received from its receiver, to maximize the probability of successful transmission.  This scenario is modeled as a stochastic game and fully distributed learning algorithms are proposed to find a correlated equilibrium (CE) and a Pareto point.  We have compared the sum of rates of all the users at the CE and the Pareto point, and we observe that the Pareto optimal power allocations provide higher probability of successful transmission.  We have also compared our algorithms with two other recent learning algorithms in literature \cite{rl_ce}, \cite{L_eff}.  The CE obtained by our algorithm performs as well as the CE obtained via the algorithm in \cite{rl_ce} but our algorithm converges much faster.  On the other hand the performance of our CE is better than the best point obtained by the algorithm in \cite{L_eff}.\par
We also note in our examples that we can achieve a higher sum rate by operating at a CCE than operating at a CE but at the expense of more number of iterations to compute it.  But, in general, it is not guaranteed that a CCE yields a better sum rate than a CE.  On the other hand a Nash bargaining solution may be a better operating point than an arbitrary Pareto point as it provides fairness among the users.\par
Transmitting at multiple rates can significantly improve the sum rate but as the strategy set of each player is enlarged, it requires more number of iterations to converge to either a Nash bargaining solution or a CCE.  In practice, which algorithm to use depends on the nature of the problem, i.e., if for example, the cardinality of the overall action space is not large, then the users can find a Pareto point more quickly than a CE or CCE.  But if the action space is large and if it requires to converge to an equilibrium quickly, one can use regret-matching to converge to a CE.


\begin{thebibliography}{9}
\bibitem{deng}
 S. Deng, T. Weber, and A. Ahrens, ``Capacity optimizing power allocation in Interference Channels,'' {\it AEU International Journal of Electronics and Communications}, Vol.63, pp. 139-147, Feb. 2009.
\bibitem{dt}
  Daniela Tuninetti, ``Gaussian Fading Interference Channels: Power Control,'' {\it Proc. of the 42nd Asilomar Conference on Signals, Systems and Computers}, Monterey, CA, pp. 701-706, October 2008.
\bibitem{PA}
K. A. Chaitanya, U. Mukherji, and V. Sharma, ``Power Allocation for Interference Channel,'' {\it Proc. of National Conference on Communications}, New Delhi, 2013.
\bibitem{palomar}
G. Scutari, D. P. Palomar, and S. Barbarossa, ``Optimal Linear Precoding Strategies for Wideband Non-Cooperative Systems Based on Game Theory-Part II: Algorithms,'' {\it IEEE Trans on Signal Processing}, Vol.56, no.3, pp. 1250-1267, March 2008.
\bibitem{Conv}
K. W. Shum, K.-K. Leung, and C. W. Sung,  ``Convergence of Iterative Waterfilling Algorithm for Gaussian Interference Channels,'' {\it IEEE Journal on Selected Areas in Comm.,} Vol.25, no.6, pp. 1091-1100, August 2007.
\bibitem{MIMO_IWF}
G. Scutari, D. P. Palomar, and S. Barbarossa, ``The MIMO Iterative Waterfilling Algorithm,'' {\it IEEE Trans on Signal Processing}, Vol. 57, No.5, May 2009.
\bibitem{async}
G. Scutari, D. P. Palomar, and S. Barbarossa, ``Asynchronous Iterative Water-Filling for Gaussian Frequency-Selective Interference Channels'', {\it IEEE Trans on Information Theory}, Vol.54, No.7, July 2008.
\bibitem{luca}
L. Rose, S. M. Perlaza, and M. Debbah, ``On the Nash Equilibria in Decentralized Parallel Interference Channels,'' {\it Proc. of International Conference on Communications}, Kyoto, 2011.
\bibitem{arxiv}
K. A. Chaitanya, U. Mukherji, and V. Sharma, ``Algorithms for Stochastic Games on Interference Channels,'' {\it Proc. of National Conference on Communications}, Mumbai, 2015.
\bibitem{game_net}
S. Lasaulce, and H. Tembine, ``Game Theory and Learning for Wireless Networks: Fundamentals and Applications,'' {\it Elsevier}, 2011.
\bibitem{r_ce}
S. Hart, and A. Mas-Colell, ``A Simple Adaptive Procedure Leading to Correlated Equilibrium,'' {\it Econometrica}, Vol. 68, No.5, pp. 1127-1150, Sept. 2000.
\bibitem{L_mixed}
D. Fudenberg, and D. M. Kreps, ``Learning Mixed Equilibria,'' {\it Games and Economic Behavior}, vol. 5, pp. 320-367, 1993.
\bibitem{R_wc}
N. Mastronarde, and Mihaela van der Schaar, ``Fast Reinforcement Learning for Energy-Efficient Wireless Communication,'' {\it IEEE Tran. on Signal Processing}, vol. 59, No. 12, Dec. 2011.
\bibitem{D_mac}
P. Mertikopoulos, E. V. Belmega, A. L. Moustakas, and S. Lasaulce, ``Distributed Learning Policies for Power Allocation in Multiple Access Channels,'' {\it IEEE Journal on Selected Areas in Communications}, Vol. 30, No. 1, pp. 96-106, Jan. 2012.
\bibitem{L_mimo}
E. V. Belgama, S. Lasaulce, M. Debbah, and A. Hjorungnes, ``Learning Distributed Power Allocation Policies in MIMO Channels,'' {\it European Signal Processing Conference}, Aalborg, Denmark, 2010.
\bibitem{D_te}
L. Rose, S. M. Perlaza, M. Debbah, and C. J. Le Martret, ``Distributed Power Allocation with SINR Constraints Using Trial and Error Learning,'' {\it IEEE Wireless Communications and Networking Conference}, Paris, April 2012.
\bibitem{S_chan}
E. Sabir, R. El-Azouzi, V. Kavitha, Y. Hayel, and E. Bouyakhf, ``Stochastic Learning Solution for Constrained Nash Equilibrium Throughput in Non Saturated Wireless Collision Channels,'' {\it Proc. the 3rd ICST/ACM International Workshop on Game Theory in Communication Networks}, Pisa, Italy, Oct. 2009.
\bibitem{D_net}
H. Jang, S. Y. Yun, J. Shin, and Y. Yi, ``Distributed Learning for Utility Maximization over CSMA-based Wireless Multihop Networks,'' {\it IEEE International Conference on Computer Communications}, Toronto, Canada, April-May 2014.
\bibitem{Q_fn}
H. Saad, A. Mohamed, and T. ElBatt, ``Cooperative Q-learning Techniques for Distributed Online Power Allocation in Femtocell Networks,'' {\it Wireless Communications and Mobile Computing}, doi: 10.1002/wcm.2470, 2014.
\bibitem{rl_ce}
S. Hart, and A. Mas-Colell, ``A Reinforcement Procedure Leading to Correlated Equilibrium,'' {\it Economics Essays}, Springer Berlin Heidelberg, 2001. 
\bibitem{tcp}
K. R. Fall, and W. R. Stevens, ``TCP/IP Illustrated Volume 1,'' Addison-Wesley, 2012.
\bibitem{L_eff}
B. S. R. Pradelski, and H. P. Young, ``Learning Efficient Nash Equilibria in Distributed Systems,'' {\it Games and Economic Behavior}, Elsevier, 75, pp.882-897, 2012.
\bibitem{L_tel}
H. P. Young, ``Learning by Trial and Error,'' {\it Games and Economic Behavior}, Elsevier, 65, pp.626-643, 2009.
\bibitem{L_rlrd}
T. Borgers, and R. Sarin, ``Learning Through Reinforcement Learning and Replicator Dynamics,'' {\it Journal of Economic Theory}, 77, pp.1-14, 1997.
\bibitem{O_agent}
W. B. Arthur, ``On Designing Economic Agents that Behave like Human Agents,'' {\it Journal of Evolutionary Economics}, Springer-Verlag, pp.1-22, 1993.
\bibitem{A_opt}
M. Zuluaga, A. Krause, G. Sergent, and M. Puschel, ``Active Learning for Multi-Objective Optimization,'' {\it Proc. 30th International Conference on Machine Learning}, Atlanta, Georgia, USA, 2013.
\bibitem{A_learn}
J. R. Marden, H. P. Young, and L. Y. Pao, ``Achieving Pareto Optimality Through Distributed Learning,'' {\it IEEE Conference on Decision and Control}, Maui, Hawaii, USA, Dec. 2012.
\bibitem{MOP}
K. Miettinen, ``Nonlinear Multiobjective Optimization,'' Kluwer Academic Publishers, 1999.
\bibitem{BASAR}
Z. Han, D. Niyato, W. Saad, T. Basar, and A. Hjorungnes, ``Game Theory in Wireless and Communication Networks,'' {\it Cambridge University Press}, 2012.
\bibitem{adapt}
S. Hart, ``Adaptive Heuristics,'' {\it Econometrica}, Vol. 73, No. 5, pp. 1401-1430, September, 2005.
\bibitem{s_s}
H. H. Hoos, and T. Stutzle, ``Stochastic Local Search: Foundations and Applications,'' {\it Morgan Kaufmann/Elsevier}, 2004.
\bibitem{w_m}
N. Littlestone, and M. K. Warmuth, ``The Weighted Majority Algorithm,'' {\it Information and Computation}, 108(2), pp 212-261, 1994.
\bibitem{m_w}
S. Arora, E. Hazan, and S. Kale, ``The Multiplicative Weights update method: a meta algorithm and applications,'' {\it Theory of Computing}, 8(1), pp 121-164, 2012.
\bibitem{plg}
N. Cesa-Bianchi, and G. Lugosi, ``Prediction, Learning, and Games,'' {\it Cambridge University Press}, 2006.
\bibitem{poa}
T. Roughgarden, ``Intrinsic Robustness of the Price of Anarchy,'' {\it ACM Symposium on Theory of Computing}, pp.513-522, 2009.
\bibitem{nash_b}
J. Nash, ``The Bargaining Problem,'' {\it Econometrica}, 18:155-162, 1950.
\bibitem{prop_fair}
F. Kelly, A. Maulloo, and D. Tan, ``Rate Control for Communication Networks: Shadow Prices, Proportional Fairness and Stability,'' {\it Journal of the Operations Research Society}, Vol. 49, No. 3, pp. 237-252, March, 1998.
\bibitem{nb_pf}
H. Boche, and M. Schubert, ``Nash Bargaining and Proportional Fairness for Wireless systems,'' {\it IEEE/ACM Transactions on Networking}, Vol. 17, No. 5, pp. 1453-1466, October, 2009.
\end{thebibliography}
\end{document}